\documentclass[11pt]{article}  

\usepackage{amssymb}
\usepackage{amsmath}
\usepackage{amsthm}
\usepackage{color}
\usepackage{colortbl}
\usepackage{graphicx}
\usepackage{hyperref}
\usepackage{fullpage}
\usepackage{xspace}

\usepackage{lmodern}
\usepackage[T1]{fontenc}

\usepackage{caption}
\usepackage{algorithm}
\usepackage{algpseudocode}
\algnewcommand{\LineComment}[1]{\State \(\triangleright\) \textit{#1}}
\algnewcommand{\Annotation}[1]{\State \textcolor{blue}{#1}}

\newtheorem{theorem}{Theorem}
\newtheorem{lemma}[theorem]{Lemma}

\newcommand{\etal}{\emph{et al.}\xspace}
\newcommand{\eps}{\epsilon}

\newcommand{\R}{\mathbb{R}}

\renewcommand{\S}{\mathcal{S}}

\newcommand{\ones}{\mathbf{1}}

\newcommand{\x}{\vec{x}}
\newcommand{\y}{\vec{y}}
\newcommand{\z}{\vec{z}}
\newcommand{\g}{\vec{g}}
\newcommand{\p}{\vec{p}}

\newcommand{\xtp}{\vec{x}^{(t+1)}}
\newcommand{\ytp}{\vec{y}^{(t+1)}}

\newcommand{\xt}{\vec{x}^{(t)}}
\newcommand{\yt}{\vec{y}^{(t)}}
\newcommand{\pt}{\vec{p}^{(t)}}
\newcommand{\ptp}{\vec{p}^{(t+1)}}
\newcommand{\xz}{\vec{x}^{(0)}}
\newcommand{\yz}{\vec{y}^{(0)}}
\newcommand{\pz}{\vec{p}^{(0)}}
\newcommand{\xone}{\vec{x}^{(1)}}
\newcommand{\yone}{\vec{y}^{(1)}}
\newcommand{\pone}{\vec{p}^{(1)}}

\newcommand{\one}[1]{\vec{1}_{#1}}
\newcommand{\norm}[1]{\left\|#1\right\|}

\newcommand{\xtdot}{\dot{\vec{x}}^{(t)}}
\newcommand{\ytdot}{\dot{\vec{y}}^{(t)}}
\newcommand{\ptdot}{\dot{\vec{p}}^{(t)}}

\newcommand{\xopt}{\vec{x}^*}

\begin{document}
\title{A Parallel Double Greedy Algorithm for Submodular Maximization}
\author{
Alina Ene\thanks{Department of Computer Science, Boston University, {\tt aene@bu.edu}.}
\and
Huy L. Nguy\~{\^{e}}n\thanks{College of Computer and Information Science, Northeastern University, {\tt hlnguyen@cs.princeton.edu}.} 
\and
Adrian Vladu\thanks{Department of Computer Science, Boston University, {\tt avladu@bu.edu}.}
}
\date{}
\maketitle

\begin{abstract}
We study parallel algorithms for the problem of maximizing a non-negative submodular function. Our main result is an algorithm that achieves a nearly-optimal $1/2 -\epsilon$ approximation using $O(\log(1/\epsilon) / \epsilon)$ parallel rounds of function evaluations. Our algorithm is based on a continuous variant of the double greedy algorithm of Buchbinder \etal that achieves the optimal $1/2$ approximation in the sequential setting. Our algorithm applies more generally to the problem of maximizing a continuous diminishing-returns (DR) function.
\end{abstract}

\section{Introduction}

In this paper, we study parallel algorithms for the problem of maximizing a submodular function. A set function $f$ on a ground set $V$ is submodular if it satisfies the following {\em diminishing return} property: $f(A\cup \{v\}) - f(A) \ge f(B\cup\{v\}) - f(B)$ for all sets $A\subseteq B$ and all elements $v\not\in B$. The problem of maximizing a submodular function is a fundamental combinatorial optimization problem that captures many problems in both theory and practice. From the theory point of view, it generalizes well-studied problems including the maximum cut and the maximum directed cut problems. From the practical point of view, it captures many applications ranging from maximum a-posteriori (MAP) inference for determinantal point processes (DPP) and mean-field inference in log-submodular models, to quadratic programming and revenue maximization in social networks~\cite{kulesza2012determinantal, gillenwater2012near,bian2017continuous,ito2016large}. 

The problem of maximizing a submodular function has received considerable attention~\cite{feige2011maximizing,buchbinder2015tight}, leading to several algorithms based on random sampling, greedy, and local search that achieve constant factor approximation guarantees for the problem. In a breakthrough work, Buchbinder \etal~\cite{buchbinder2015tight} introduced the double greedy (also known as bi-directional greedy) algorithm, a very elegant algorithm that achieves a $1/2$ approximation, which is optimal in the value oracle model~\cite{feige2011maximizing}.

A significant drawback of greedy and local search algorithms is that they are inherently sequential and adaptive. The \emph{adaptivity} of an algorithm is the number of sequential rounds of queries it makes to the evaluation oracle of the function, where in every round the algorithm is allowed to make polynomially-many parallel queries. Motivated by applications in a wide-range of domains, Balkanski and Singer~\cite{BS18} initiated the study of adaptivity (or parallelization) for submodular maximization problems. The work~\cite{BS18} considered the problem of maximizing a monotone submodular function subject to a cardinality constraint, and gave an $1/3 - \eps$ approximation algorithm using $O(\log{n}/\eps^2)$ rounds of adaptivity as well as a hardness result showing that $\Omega(\log{n}/\log\log{n})$ rounds of adaptivity are necessary to obtain a $\Omega(1/\log{n})$ approximation. A recent line of work studies the tradeoff between approximation guarantee and adaptivity for both monotone and non-monotone submodular maximization problems subject to cardinality, packing, and matroid constraints~\cite{BS18,EN18,BRS18,FMZ18,CQ18,BBY18,balkanski2018optimal,ENV18}.

The work~\cite{ENV18} implies a $1/e - \eps$ approximation using $O(\log{n}/\eps^2)$ rounds for maximizing a submodular function as a special case. The random sampling algorithm of \cite{feige2011maximizing} achieves a $1/4$ approximation using one round of adaptivity: the algorithm returns a random set that includes each element independently at random with probability $1/2$ and never evaluates the function. In summary, for the problem of maximizing a submodular function, we can obtain a $1/4$ approximation in 1 adaptive round, a $1/e - \eps$ approximation in $\Theta(\log{n}/\eps^2)$ adaptive rounds, and a $1/2$ approximation in $\Theta(n)$ adaptive rounds.

\medskip
\textbf{Our contribution.} In this paper, we show that we can obtain a nearly-optimal approximation guarantee using a constant number of adaptive rounds. More precisely, we give a parallel algorithm that achieves a $1/2 - \eps$ approximation using $O(\log(1/\eps)/\eps)$ adaptive rounds. Our parallel algorithm is based on a continuous variant of the double greedy algorithm of Buchbinder \etal~\cite{buchbinder2015tight}. Our algorithm applies more generally to the problem of maximizing a continuous diminishing-returns (DR) submodular function.  Recent work has shown that DR-submodular optimization problems have applications beyond submodular maximization \cite{bian2016guaranteed,bian2017continuous,soma2017non,bian2018optimal}, including several of the applications mentioned above. 

\begin{theorem}
For every $\eps > 0$, there is an algorithm for maximizing a DR-submodular function $f: [0, 1]^n \rightarrow \R_+$ with the following guarantees:
\begin{itemize}
\item The algorithm is deterministic if provided oracle access for evaluating $f$ and its gradient $\nabla f$;
\item The algorithm achieves an approximation guarantee of $\frac{1}{2} -\eps$;
\item The number of rounds of adaptivity and evaluations of $f$ and $\nabla f$ are $O\left(\frac{\log(1/\eps)}{\eps} \right)$.
\end{itemize}
\end{theorem}

\medskip\noindent
{\bf Related work.} The same result was obtained independently by Chen, Feldman, and Karbasi~\cite{CFK18}.

\subsection{Our techniques}

Similar to the double greedy algorithm for the sequential setting, our algorithm maintains two solutions $\vec{x}$ and $\vec{y}$ and iteratively makes them more similar over time. In contrast with the sequential algorithm where the coordinates are fixed one by one, our algorithm updates many coordinates in parallel in each iterations based on the gradient of $f$ at $\vec{x}$ and $\vec{y}$. Intuitively the gradients give an upper bound on the {\em potential} gain we can get from changing the coordinates of $\vec{x}$ and $\vec{y}$. Our algorithm works following the potential function that measure exactly this amount:
\[ \Phi = \langle \nabla f(\x) - \nabla f(\y), \one{\mathcal{S}}\rangle,\]
where $\mathcal{S}$ is the set of all coordinates $i$ for which $\nabla_i f(\vec{x}) > 0$ and $\nabla_i f(\vec{y}) < 0$. The sum of the positive coordinates of $\nabla f(\vec{x})$ is an upper bound on how much $f(\vec{x})$ can increase by increasing $\vec{x}$. The sum of the negative coordinates of $\nabla f(\vec{y})$ is an upper bound on how much $f(\vec{y})$ can increase by decreasing $\vec{y}$. The algorithm increases $\vec{x}$ and decreases $\vec{y}$ in iterations until either they meet or the maximum potential gain becomes too small. At that point, we can return $\vec{x}$ as our approximate solution. Note that during the course of execution, by submodularity, the potential can never increase.

Let $M = \max_{\vec{z}\in[0,1]^n} f(\vec{z})$. The main part of the execution happens during the time where the potential goes from $M/\eps$ down to $\eps M$. As mentioned before, once the potential drops below $\eps M$, the algorithm finishes as there is not much more room for improvement. In addition to $\vec{x}$ and $\vec{y}$, the analysis maintains the projection of the optimal solution $\vec{x}^*$ to the box defined by $\vec{x}$ and $\vec{y}$ i.e. $\vec{p} = (\vec{x}^* \wedge \vec{y})\vee \vec{x}$. For each iteration, we need to analyze the gain in $(f(\vec{x})+f(\vec{y}))/2$ and the loss in $f(\vec{p})$ and show that the gain is at least as large as the loss. Consider a coordinate $i$. If $(\nabla f(\vec{x}))_i \le 0$ then, by submodularity, $(\nabla f(\vec{z}))_i \le 0$ for all $\vec{z} \ge \vec{x}$. Thus, we can immediately reduce $y_i$ to $x_i$ and in the process, increase the value of $f(\vec{y})$. This step also potentially lowers $p_i$ to $x_i$ but since the gradient is negative, this step also increases the value of $f((\vec{x}^* \wedge \vec{y}) \vee \vec{x})$. A similar argument works for the case $(\nabla f(\vec{y}))_i \ge 0$. Thus, the interesting coordinates are the set $S$ of coordinates $i$ where $(\nabla f(\vec{y}))_i < 0$ and $(\nabla f(\vec{x}))_i > 0$. For each coordinate $i \in S$, our algorithm increases $x_i$ and decreases $y_i$ proportional to the corresponding gradient entries: the increase in $x_i$ is $\eta \nabla_i f(\x) / (\nabla_i f(\x) - \nabla_i f(\y))$ and the decrease in $y_i$ is $\eta (- \nabla_i f(\y)) / (\nabla_i f(\x) - \nabla_i f(\y))$. The step size $\eta$ is chosen so that the potential remains roughly the same as before up to a $1 - \eps$ factor. The effect of this step is that the values of $f(\vec{x})$ and $f(\vec{y})$ go up but the value of $f(\vec{p})$ might decrease (because we need to project the optimal solution $\vec{x}^*$ to a smaller box). It turns out that one can relate these changes and argue that the gain outweighs the loss (Lemma~\ref{lem:disc-invariant}). Thus, at the end of the algorithm, the values of $f(\vec{x}), f(\vec{y}), f(\vec{p})$ are close to each other and because the total gain of going from $(f(\vec{0})+f(\vec{1}))/2$ to $(f(\vec{x})+f(\vec{y}))/2$ outweighs the loss of going from $f(\vec{x}^*)$ to $f(\vec{p})$,  the values of $f(\vec{x}), f(\vec{y})$ are at least $f(\vec{x}^*)/2$.

To analyze the number of iterations, we show that each iteration decreases the potential by a $1 - \eps$ factor. Thus the potential decreases from its initial value to $\eps M$. To make sure that the number of iterations is small, we need to start the process from a point where the gradient is not too large i.e. $\|\nabla f(\vec{x})\vee \vec{0}\|_1 \le M/\eps$. It turns out that there is a simple solution: we start from $\vec{x} = \eps \vec{1}$. Because $f(\vec{0})\ge 0$, $f(\eps \vec{1}_T)\le M~\forall T\subseteq V$ and the diminishing return property, we have $\sum_{e\in T} (\nabla f(\eps \vec{1}))_e \le M/\eps\,$ which implies $\|\nabla f(\vec{x})\vee \vec{0}\|_1 \le M/\eps$. Starting from this point does not result in a significant loss in the optimal value because $f(\vec{x}^* \vee \vec{x}) \ge (1-\|\vec{x}\|_{\infty}) f(\vec{x}^*)~\forall \vec{x}$. Thus, the potential decreases from $M/\eps$ to $\eps M$ in $O(\ln(1/\eps)/\eps)$ iterations.

\medskip
{\bf Paper outline.}
In Section~\ref{sec:contdoublegreedy}, we describe and analyze a continuous variant of the sequential double greedy algorithm of Buchbinder \etal~\cite{buchbinder2015tight} that updates many coordinates simultaneously. We build on this algorithm and analysis in Section~\ref{sec:paralleldoublegreedy}, and obtain our parallel double greedy algorithm.

\section{Preliminaries}

Let $f: [0, 1]^n \rightarrow \R_+$ be a non-negative function. The function is \emph{diminishing returns submodular} (DR-submodular) if $\forall \vec{x} \leq \vec{y} \in [0, 1]^n$ (where $\leq$ is coordinate-wise), $\forall i \in [n]$, $\forall \delta \in [0, 1]$ such that $\vec{x} + \delta \vec{1}_{\{i\}}$ and $\vec{y} + \delta \vec{1}_{\{i\}}$ are still in $[0, 1]^n$, it holds
  \[f(\vec{x} + \delta \vec{1}_{\{i\}}) - f(\vec{x}) \geq f(\vec{y} + \delta \vec{1}_{\{i\}}) - f(\vec{y}),\]
where $\vec{1}_{\{i\}}$ is the $i$-th basis vector, i.e., the vector whose $i$-th entry is $1$ and all other entries are $0$.

If $f$ is differentiable, $f$ is DR-submodular if and only if $\nabla f(\vec{x}) \geq \nabla f(\vec{y})$ for all $\vec{x} \leq \vec{y} \in [0, 1]^n$. If $f$ is twice-differentiable, $f$ is DR-submodular if and only if all the entries of the Hessian are \emph{non-positive}, i.e., $\frac{\partial^2 f}{\partial x_i \partial x_j}(\vec{x}) \leq 0$ for all $i, j \in [n]$. 

For simplicity, throughout the paper, we assume that $f$ is differentiable. We assume that we are given black-box access to an oracle for evaluating $f$ and its gradient $\nabla f$. We extend the function $f$ to $\R^n_+$ as follows: $f(\vec{x}) = f(\vec{x} \wedge \vec{1})$, where $(\vec{x} \wedge \vec{1})_i = \min\{x_i, 1\}$.

The multilinear extension of a submodular set function is DR-submodular~\cite{Calinescu2011,Vondrak2008}. A fractional solution to the problem of maximizing the multilinear extension can be rounded without any loss: given $\x \in [0, 1]^n$, round up each coordinate $i$ independently at random with probability $x_i$.

\medskip
{\bf Basic notation.} Let $V$ be a finite ground set of size $n = |V|$; without loss of generality, $V = \{1, 2, \dots, n\} = [n]$. We use e.g. $\vec{x} = (x_1, \dots, x_n)$ to denote a vector in $\R^n$. 
We use the following vector operations: $\vec{x} \vee \vec{y}$ is the vector whose $i$-th coordinate is $\max\{x_i, y_i\}$; $\vec{x} \wedge \vec{y}$ is the vector whose $i$-th coordinate is $\min\{x_i, y_i\}$; $\vec{x} \circ \vec{y}$ is the vector whose $i$-th coordinate is $x_i \cdot y_i$. We write $\vec{x} \leq \vec{y}$ to denote that $x_i \leq y_i$ for all $i \in [n]$. Let $\vec{0}$ (resp. $\vec{1}$) be the $n$-dimensional all-zeros (resp. all-ones) vector. Let $\vec{1}_S \in \{0, 1\}^V$ denote the indicator vector of $S \subseteq V$, i.e., the vector that has a $1$ in entry $i$ if and only if $i \in S$. Similarly given a vector $\vec{x}$, we let $\vec{1}_{\vec{x}}$ be the indicator vector for strictly positive elements of $\vec{x}$.

We will use the following result that was shown in previous work~\cite{ChekuriJV15}.

\begin{lemma}[{\cite[Lemma~7]{ChekuriJV15}}]
\label{lem:x-or-opt}
Let $f: [0, 1]^n \rightarrow \R_+$ be a DR-submodular function. For all $\vec{x}^* \in [0, 1]^n$ and $\vec{x} \in [0, 1]^n$, $f(\vec{x}^* \vee \vec{x}) \geq (\vec{1} - \|\vec{x}\|_{\infty}) f(\vec{x}^*)$.
\end{lemma}

The following result follows from concavity in non-negative directions.

\begin{lemma}
\label{lem:concavity}
Let $f: [0, 1]^n \rightarrow \R_+$ be a DR-submodular function. For all $\vec{x} \leq \vec{y}$,
\[ \left< \nabla f(\vec{x}), \vec{y} - \vec{x} \right> \geq f(\vec{y}) - f(\vec{x}) \geq \left< \nabla f(\vec{y}), \vec{y} - \vec{x} \right> \]
\end{lemma}

\section{Continuous Double Greedy Dynamics}
\label{sec:contdoublegreedy}

%In this section, we introduce and analyze several continuous double greedy algorithms. 

Given two points $\xt \leq \yt$ we define $\pt=\textnormal{Proj}_{[\xt,\yt]}x^{*}$.
We initialize the algorithm with $\xz=\vec{0}$, $\yz=\vec{1}$,
which means that $\pz=\xopt$. Throughout the algorithm we update
$\xt$ and $\yt$ such that the following invariant holds for
every $t$:
\begin{align}
\frac{d}{dt}\left(\frac{1}{2}\left(f(\xt)+f(\yt)\right)+\alpha\cdot f(\pt)\right) & \geq0\nonumber 
\end{align} 
or equivalently
\begin{align}
\frac{1}{2} \left(\langle\nabla f(\xt),\xtdot\rangle+\langle\nabla f(\yt),\ytdot\rangle \right) +\alpha\cdot\langle\nabla f(\pt),\ptdot \rangle & \geq0\label{eq:inv}
\end{align}

\begin{lemma}
\label{lem:bound-alpha}
Consider a continuous trajectory for $(\xt,\yt)_{0\leq t\leq1}$ such that at all times $\xt\leq \yt$, and $\xone = \yone$. If the invariant from (\ref{eq:inv}) holds for al $t$, then $f(\xt)\geq\frac{\alpha}{1+\alpha}f(\xopt)$.
\end{lemma}

\begin{proof}
We consider the total gain in function value i.e. $$\frac{1}{2}\left(f(\xone)-f(\xz)+f(\yone)-f(\yz)\right)$$
and compare it to the total drop in function value for the projected
optimum i.e. $f(\pz)-f(\pone).$ By integrating (\ref{eq:inv})
we obtain
\[
\frac{1}{2}\left(f(\xone)-f(\xz)+f(\yone)-f(\yz)\right)\geq\alpha\cdot\left(f(\pz)-f(\pone)\right)
\]
Since $\xone =\yone =\pone :=\vec{x}$, and $\pz=\xopt$ we have
\begin{align*}
\frac{1}{2}\left(2f(\vec{x})-f(\xz)-f(\yz)\right) & \geq\alpha\cdot\left(f(\xopt)-f(\vec{x})\right)\\
(1+\alpha)f(\vec{x}) & \geq\frac{1}{2}\left(f(\xz)+f(\yz)\right)+\alpha f(\xopt)\\
f(\vec{x}) & \geq\frac{1}{2\alpha}\left(f(\xz)+f(\yz)\right)+\frac{\alpha}{1+\alpha}f(\xopt)\geq\frac{\alpha}{1+\alpha}f(\xopt)
\end{align*}
\end{proof}

We now describe two strategies that enforce (\ref{eq:inv}) with $\alpha = 1$, and thus they yield a $1/2$ approximation. The first strategy can be viewed a continuous version of the Buchbinder \etal discrete double greedy algorithm. Our parallel algorithm that we give in Section~\ref{sec:paralleldoublegreedy} is a discretization of this continuous dynamic. A key difference between this continuous dynamic (and its corresponding discretization given in Section~\ref{sec:paralleldoublegreedy}) is that it updates many coordinates simultaneously, whereas the discrete double greedy algorithm of Buchbinder \etal updates only one coordinate at a time.

\begin{lemma}
\label{lem:dg-inv}
The following update rule preserves (\ref{eq:inv}) with $\alpha = 1$. For every coordinate $i$ such that $\nabla_i f(\xt) > 0$ and $\nabla_i f(\yt) < 0$, we set
\begin{align*}
  \xtdot &= \frac{\nabla_i f(\xt)}{\nabla_i f(\xt) - \nabla_i f(\yt)}\\
  \ytdot &= \frac{\nabla_i f(\yt)}{\nabla_i f(\xt) - \nabla_i f(\yt)}
\end{align*}
\end{lemma}
\begin{proof}
We show that the invariant (\ref{eq:inv}) is maintained for every coordinate in turn. Consider a coordinate $i$. If $\nabla f(\xt) \leq 0$ or $\nabla f(\yt) \geq 0$, we have $\xtdot = \ytdot = \ptdot = 0$, and the invariant holds. Therefore we may assume that $\nabla f(\xt) > 0$ and $\nabla f(\yt) < 0$. We have
\begin{align*}
\nabla_i f(\xt) \xtdot &= \frac{(\nabla_i f(\xt))^2}{\nabla_i f(\xt) - \nabla_i f(\yt)}\\
\nabla_i f(\yt) \ytdot &= \frac{(\nabla_i f(\yt))^2}{\nabla_i f(\xt) - \nabla_i f(\yt)}
\end{align*}
We now analyze $\nabla_i f(\pt) \ptdot$. Note that $\xt_i$ increases and $\yt_i$ decreases. Additionally, $\pt_i$ changes only if it is equal to $\xt_i$ or $\yt_i$. We consider each of these cases in turn:
\begin{itemize}
\item $\pt_i = \xt_i$. In this case, we have $\ptdot_i = \xtdot_i > 0$. Since $\xt \leq \pt \leq \yt$, we have $\nabla f(\xt) \geq \nabla f(\pt) \geq \nabla f(\yt)$. Thus
\[ \nabla_i f(\pt) \ptdot_i 
  = \nabla_i f(\pt) \xtdot_i
  \geq \nabla_i f(\yt) \xtdot_i
  = \frac{\nabla_i f(\xt) \nabla_i f(\yt)}{\nabla_i f(\xt) - \nabla_i f(\yt)}
\]
\item $\pt_i = \yt_i$. In this case, we have $\ptdot_i = \ytdot_i < 0$. Since $\xt \leq \pt \leq \yt$, we have $\nabla f(\xt) \geq \nabla f(\pt) \geq \nabla f(\yt)$. Thus
\[ \nabla_i f(\pt) \ptdot_i 
  = \nabla_i f(\pt) \ytdot_i
  \geq \nabla_i f(\xt) \ytdot_i
  = \frac{\nabla_i f(\xt) \nabla_i f(\yt)}{\nabla_i f(\xt) - \nabla_i f(\yt)}
\]
Therefore
\[
\frac{1}{2} \left(\nabla_i f(\xt) \xtdot_i + \nabla_i f(\yt) \ytdot_i \right) + \nabla_i f(\pt) \ptdot_i \geq \frac{1}{2} \frac{(\nabla_i f(\xt) + \nabla_i f(\yt))^2}{\nabla_i f(\xt) - \nabla_i f(\yt)} \geq 0
\]
\end{itemize}
\end{proof}

The following strategy is also very natural and we can analyze using a similar proof (see the appendix). 

\begin{lemma}
\label{lem:dg-inv2}
Setting $\xtdot = \nabla f(\xt)^+$, $\ytdot=\nabla f(\xt)^-$ preserves (\ref{eq:inv}), with $\alpha=1$.
\end{lemma}

It is a simple observation that if the dynamic stops before making
$\xt = \yt$ we can return any of the two points.

\section{Parallel Double Greedy Algorithm}
\label{sec:paralleldoublegreedy}
The discrete parallel version of double greedy is based on the continuous method described in Section~\ref{sec:contdoublegreedy} (we consider the update rule analyzed in Lemma~\ref{lem:dg-inv}). The key point is that instead of taking infinitesimally small steps, the updates in $\vec{x}$ and $\vec{y}$ are simultaneously scaled by the largest possible step size until the first order approximation of the gain from the average of the new points fails to approximate the gain anticipated by the average of the old points within a factor of $1-\epsilon$. This is precisely captured by the condition specified by the algorithm on line~\ref{line:line-search}.
\begin{figure}[t]
\begin{algorithmic}[1]
\Procedure{ParallelDoubleGreedy($f, M$)}{}
\State $\x \gets \eps \ones$, $\y \gets (1-\eps)\ones$
\While {$\langle \nabla f(\x)-\nabla f(\y) , \y-\x  \rangle \geq \epsilon M$}
\State $\S \gets \{i: \nabla_i f(\x) > 0 \textnormal{ and } \nabla_i f(\y) < 0\}$
\ForAll {$i \notin \mathcal{S}$ : $\x_i < \y_i$ }
\If {$\nabla f(\x)_i \leq 0$}
\State $\y_i \gets \x_i$ 
\Else
\State $\x_i \gets \y_i$
\EndIf
\EndFor
\State $\Delta \x \gets \vec{0}, \Delta \y \gets \vec{0}$
\ForAll{$i \in \mathcal{S}$}
\State $(\Delta \x)_i = \frac{\nabla_i f(\x)}{\nabla_i f(\x) - \nabla_i f(\y)}$
\State $(\Delta \y)_i = \frac{\nabla_i f(\y)}{\nabla_i f(\x) - \nabla_i f(\y)}$
\EndFor
\State Line search for largest $\eta > 0$ such that  
\begin{align*}
f(\x + \eta \Delta \x)-f(\x) + f(\y+\eta \Delta \y)-f(\y)
\\
\geq (1-\epsilon) \left(\langle \nabla f(\x), \eta \Delta\x \rangle + \langle \nabla f(\y), \eta \Delta\y \rangle \right)
\end{align*} \label{line:line-search} 
\State $\x \gets \x + \eta \Delta \x$, $\y \gets \y + \eta \Delta \y$
\EndWhile
\State \Return $\arg\max \{ f(\x), f(\y) \}$
\EndProcedure
\end{algorithmic}

\caption{Description of our parallel algorithm for non-monotone submodular maximization. The line search is performed approximately using $O(1)$ parallel rounds, each with $O(\log(1/\eps)/\eps)$ queries.}
\label{fig:algo}
\end{figure}

\smallskip
{\bf Implementation of the line search.}
We perform each line search approximately as follows. Fix an iteration of the algorithm and let $\eta^*$ be the optimal step for the line search on line~\ref{line:line-search}. We first check whether the step size $\eta = \eps^{O(1)}$ meets the condition, where the $O(1)$ is a sufficiently large constant (a constant of $4$ will suffice for us). If this step size does not meet the condition, then we use this step size and finish the search. In the following, we assume that this step size meets the condition. We show that, for any constant $c \geq 1$, we can find the minimum power of $(1 - \eps^c)$ that exceeds $\eta^*$ using $O(c)$ parallel rounds, each of which performs $O(\log(1/\eps)/\eps)$ parallel queries. The first round finds the minimum power of $(1 - \eps)$ that exceeds $\eta^*$, and subsequent rounds refine the approximation. In the first round, we try $\eta = 1, (1 - \eps), (1 - \eps)^2, \dots, \eps^{O(1)}$ in parallel and take the minimum $\eta$ that fails the condition. We iteratively refine this approximation so that, after $j$ rounds, we have an integer $i_j$ such that $\eta^* \in [(1 - \eps^{j})^{i_j}, (1 - \eps^j)^{(i_j - 1)})$. Given $i_{j}$, we find the integer $i_{j + 1}$ by trying all the powers of $1-\eps^{j+1}$ in the range $[(1 - \eps^{j})^{i_j}, (1 - \eps^j)^{(i_j - 1)})$. 

It suffices for our purposes to find an approximate line step $\eta = (1 - \eps^4)^{i}$, where $i$ is such that $\eta^* \in [(1 - \eps^4)^{i}, (1 - \eps^4)^{(i - 1)})$. We can obtain such an approximation using $4$ parallel rounds, each of which performs $O(\log(1/\eps)/\eps)$ queries. There are other obvious tradeoffs between rounds and number of queries and we only exhibit one possible choice.

For simplicity, in the remainder of the analysis we assume that the line search is performed exactly, as the total error incurred from the approximate line searches can be bounded by $O(\eps M)$. As shown in Theorem~\ref{thm:iterations}, the relevant gradients have $\ell_1$-norm at most $M / \eps$ and the number of iterations is $O(\log(1/\eps)/\eps)$. This allows us to extend the analysis at a loss in the approximation of $O(\eps^3 M)$ per iteration, and thus $O(\eps M)$ overall. 

\smallskip
{\bf Analysis of the approximation guarantee.}
We proceed similarly to the analysis from Section~\ref{sec:contdoublegreedy}. In the following, we use $\xt$ and $\yt$ to denote the vectors $\x$ and $\y$ at the beginning of iteration $t$ of the algorithm, and similarly for the other quantities of interest. We let $\pt = \textnormal{Proj}_{[\xt, \yt]} \xopt$. We first show that the algorithm maintains the following invariant:

\begin{lemma}
\label{lem:disc-invariant}
The algorithm maintains the invariant
\[ \frac{1}{2} \left(
f(\x^{(t+1)}) - f(\xt) + f(\y^{(t+1)}) - f(\yt)
\right)
+ (1-\epsilon) \left( f(\ptp)-f(\pt) \right)
\geq 0
\]
\end{lemma}
\begin{proof}
By the choice of $\eta$, we have
\[
  \frac{1}{2} \left( f(\x^{(t+1)}) - f(\xt) + f(\y^{(t+1)}) - f(\yt) \right)
  \geq \frac{1 - \eps}{2} \eta \left(\langle \nabla f(\xt), \Delta \xt \rangle + \langle \nabla f(\yt), \Delta\yt \rangle \right)
\]
We now lower bound $f(\ptp) - f(\pt)$. We have
\begin{align*}
f(\ptp)-f(\pt) = \int_0^1 \langle \nabla f( (1-\alpha) \pt + \alpha \ptp ), \ptp-\pt \rangle d\alpha
\end{align*}
We write in shorthand $\g_\alpha = (1-\alpha) \pt + \alpha \ptp $, and $\Delta \pt = \ptp - \pt$. Now consider the coordinates of $\Delta \pt$ and partition them into two sets, one where they are positive, and one where they are negative. For the former, we lower bound the contribution of the integral
\begin{align*}
\int_0^1 \nabla_i f(\g_\alpha) (\Delta \pt)_i d\alpha \geq \nabla_i f(\yt) \cdot (\Delta \xt)_i \cdot \eta = \frac{ \nabla_i f(\xt) \nabla_i  f(\yt)}{\nabla_i f(\xt) - \nabla_i f(\yt)} \cdot \eta
\end{align*}
For the latter we similarly write
\begin{align*}
\int_0^1 \nabla_i f(\g_\alpha) (\Delta \pt)_i d\alpha \geq \nabla_i f(\xt) \cdot (\Delta \yt)_i \cdot \eta = \frac{ \nabla_i f(\xt) \nabla_i  f(\yt)}{\nabla_i f(\xt) - \nabla_i f(\yt)} \cdot \eta
\end{align*}
Therefore
\[ f(\ptp) - f(\pt) \geq \eta \sum_{i \in \mathcal{S}^{(t)}}  \frac{ \nabla_i f(\xt) \nabla_i  f(\yt)}{\nabla_i f(\xt) - \nabla_i f(\yt)} \]
It follows that
\begin{align*}
&\frac{1}{2} \left(f(\x^{(t+1)}) - f(\xt) + f(\y^{(t+1)}) - f(\yt)\right)
+ (1-\epsilon) \left( f(\ptp)-f(\pt) \right)\\
&\quad \geq \frac{(1 - \eps) \eta}{2} \sum_{i \in \mathcal{S}^{(t)}} \frac{(\nabla_i f(\xt) + \nabla_i f(\yt))^2}{\nabla_i f(\xt) - \nabla_i f(\yt)}\\
&\quad \geq 0
\end{align*}
\end{proof} 

\begin{theorem}
Given a guess $M$ for the optimal value, the algorithm described in Figure~\ref{fig:algo} returns a point $\vec{x}$ satisfying 
$f(\x) \geq \left(\frac{1}{2}-O(\epsilon)\right) f(\xopt) - \epsilon M$.
\end{theorem}
\begin{proof}
Let $\vec{x}^{(T)}$, and $\vec{y}^{(T)}$ be the last iterates produced by the algorithm. Using the invariant from Lemma~\ref{lem:disc-invariant} and summing up over all iterates we obtain that
\[
\frac{1}{2}\left( f(\x^{(T)}) - f(\xz) + f(\y^{(T)}) - f(\yz) \right) + (1-\epsilon) \left( f(\vec{p}^{(T)}) - f(\pz) \right) \geq 0
\]
Since the function is non-negative, we obtain
\[
\frac{1}{2} \left( f(\x^{(T)}) + f(\y^{(T)}) \right) \geq (1 - \eps) (f(\pz) - f(\p^{(T)})) 
\]
We now show that the stopping condition of the while loop implies that any point inside the box $[\x^{(T)}, \y^{(T)}]$ can further increase the value by at most $\epsilon M$ over the best point on the boundary of the box. Indeed, for any point $\z$ such that $\x^{(T)} \leq \z \leq \y^{(T)}$, we have
\begin{align*}
f(\z) &\leq f(\x^{(T)}) + \langle \nabla f(\x^{(T)}), \z - \x^{(T)} \rangle\\
&\leq f(\x^{(T)}) + \langle \nabla f(\x^{(T)}), \y^{(T)} - \x^{(T)} \rangle\\
&\leq f(\x^{(T)}) + \eps M
\end{align*}
Note that, in particular, $f(\y^{(T)}) \leq f(\x^{(T)}) + \eps M$. By combining with the inequality we obtained from the invariant,
\begin{align*}
&\frac{1}{2} \left( f(\x^{(T)}) + f(\y^{(T)}) \right) \geq (1 - \eps) (f(\pz) - f(\p^{(T)}))\\ 
&\Rightarrow f(\x^{(T)}) + \frac{1}{2} \eps M \geq (1 - \eps) (f(\pz) - f(\x^{(T)}) - \eps M)\\
&\Rightarrow f(\x^{(T)}) \geq \frac{1}{2} (1 - \eps) f(\pz) - \eps M
\end{align*}
Applying Lemma~\ref{lem:x-or-opt} twice (once forward and once backward) we see that $f(\pz) \geq (1-\epsilon)^2 f(\xopt)$, and the theorem follows.
\end{proof}

\smallskip
{\bf Analysis of the number of iterations.} We now show that the algorithm terminates in $O(\log(1/\epsilon)/\epsilon)$ iterations.

\begin{theorem}
\label{thm:iterations}
The algorithm terminates after $O(\log(1/\epsilon)/\epsilon)$ iterations of the main loop.
\end{theorem}
\begin{proof}
The argument is based on analyzing a potential function
\[ \Phi^{(t)} = \langle \nabla f(\xt) - \nabla f(\yt), \one{\mathcal{S}^{(t)}} \rangle \]
We show that the stopping condition for line search guarantees that this function must decrease fast.

The stopping condition for line search guarantees that 
\begin{align*}
&\langle  \nabla f(\xtp), \Delta \xt \rangle + \langle \nabla f(\ytp), \Delta \yt \rangle \\
& \leq (1-\epsilon)
\left(
\langle  \nabla f(\xt), \Delta \xt \rangle +  \langle  \nabla f(\yt), \Delta \yt \rangle
\right)
\end{align*}
By rearranging, we obtain
\begin{align*}
& \langle \nabla f(\xt) - \nabla f(\xtp), \Delta \xt \rangle + \langle \nabla f(\ytp) - \nabla f(\yt), - \Delta \yt \rangle\\
&\quad \geq \eps \langle \nabla f(\xt), \Delta \xt \rangle + \eps \langle - \nabla f(\yt), - \Delta \yt \rangle
\end{align*}
By plugging in the update rule, we obtain
\begin{align*}
& \sum_{i \in \mathcal{S}^{(t)}} (\nabla_i f(\xt) - \nabla_i f(\xtp)) \cdot \frac{\nabla_i f(\xt)}{\nabla_i f(\xt) - \nabla_i f(\yt)}\\
&\quad + \sum_{i \in \mathcal{S}^{(t)}} (\nabla_i f(\ytp) - \nabla_i f(\yt)) \cdot \left( - \frac{\nabla_i f(\yt)}{\nabla_i f(\xt) - \nabla_i f(\yt)} \right)\\
&\quad \geq \eps \sum_{i \in \mathcal{S}^{(t)}} \frac{(\nabla_i f(\xt))^2 + (\nabla_i f(\yt))^2}{\nabla_i f(\xt) - \nabla_i f(\yt)}
\end{align*}

Letting $a = \nabla_i f(\xt)$, $b = - \nabla_i f(\yt)$, and using the inequality $a^2 + b^2 \geq (a + b)^2 / 2$, we obtain
\[ \eps \sum_{i \in \mathcal{S}^{(t)}} \frac{(\nabla_i f(\xt))^2 + (\nabla_i f(\yt))^2}{\nabla_i f(\xt) - \nabla_i f(\yt)} \geq \frac{\eps}{2} \sum_{i \in \mathcal{S}^{(t)}} (\nabla_i f(\xt) - \nabla_i f(\yt)) \] 

Using that $\nabla_i f(\xt) - \nabla_i f(\xtp) \geq 0$, $\frac{\nabla_i f(\xt)}{\nabla_i f(\xt) - \nabla_i f(\yt)} \in [0, 1]$, $\nabla_i f(\ytp) - \nabla_i f(\yt) \geq 0$, and $- \frac{\nabla_i f(\yt)}{\nabla_i f(\xt) - \nabla_i f(\yt)} \in [0, 1]$, we obtain     
\begin{align*}
&\sum_{i \in \mathcal{S}^{(t)}} (\nabla_i f(\xt) - \nabla_i f(\xtp)) + \sum_{i \in \mathcal{S}^{(t)}} (\nabla_i f(\ytp) - \nabla_i f(\yt)) \\  
&\quad \geq \sum_{i \in \mathcal{S}^{(t)}} (\nabla_i f(\xt) - \nabla_i f(\xtp)) \cdot \frac{\nabla_i f(\xt)}{\nabla_i f(\xt) - \nabla_i f(\yt)}\\
&\qquad + \sum_{i \in \mathcal{S}^{(t)}} (\nabla_i f(\ytp) - \nabla_i f(\yt)) \cdot \left( - \frac{\nabla_i f(\yt)}{\nabla_i f(\xt) - \nabla_i f(\yt)} \right)\\
\end{align*}

Therefore
\[
\sum_{i \in \mathcal{S}^{(t)}} (\nabla_i f(\xt) - \nabla_i f(\xtp)) + \sum_{i \in \mathcal{S}^{(t)}} (\nabla_i f(\ytp) - \nabla_i f(\yt))  
\geq \frac{\eps}{2} \sum_{i \in \mathcal{S}^{(t)}} (\nabla_i f(\xt) - \nabla_i f(\yt)) 
\]
By rearranging, we obtain
\[ \langle \nabla f(\xtp) - \nabla f(\ytp), \one{\mathcal{S}^{(t)}} \rangle\leq \left( 1 - \frac{\eps}{2}\right) \langle \nabla f(\xt) - \nabla f(\yt), \one{\mathcal{S}^{(t)}}\rangle \]
This shows that the potential function drops multiplicatively by $1-O(\epsilon)$ per iteration.

Finally, we need to understand the range of the potential function considered during the execution of the algorithm. More precisely, we argue that the potential only decreases from $M/\eps$ to $\eps M$.
First we see that the points at which we start the process $\xz = \epsilon \ones{}$, $\yz = (1-\epsilon)\ones{}$ do not have large gradients. This is because $f(0)\ge 0$, $f(\eps \vec{1}_T)\le M~\forall T\subseteq V$ and the diminishing return property, we have $$\sum_{e\in T} (\nabla f(\eps \vec{1}))_e \le M/\eps\,$$ which implies $\|\nabla f(\xz)\vee \vec{0}\|_1 \le M/\eps$. A similar argument holds for $\y$.

Finally, we show that the stopping condition of the while loop guarantees that the potential function never decreases below $\epsilon M$. Indeed, for every iteration $t$, we have
\[ \langle \nabla f(\xt) - \nabla f(\yt), \one{\mathcal{S}^{(t)}} \rangle \geq \langle \nabla f(\xt) - \nabla f(\yt), \yt - \xt \rangle \geq \eps M,\]
since, for all $i \in \mathcal{S}^{(t)}$, we have $\nabla_i f(\xt) - \nabla_i f(\yt) \geq 0$ and $\yt_i - \xt_i \in [0, 1]$.

Hence the number of multiplicative decreases of the potential function is at most $O(\log(1/\epsilon)/\epsilon)$.
\end{proof}

%\newpage
\bibliographystyle{abbrv}
\bibliography{submodular}

\appendix

\section{Proof of Lemma~\ref{lem:dg-inv2}}

Consider the update rule in Lemma~\ref{lem:dg-inv2}. We have
\begin{align*}
\langle \nabla f(\xt), \xtdot \rangle &= \langle \nabla f(\xt), \nabla f(\xt)^+ \rangle = \norm{\nabla f(\xt)^+}_2^2 \\
\langle \nabla f(\yt), \ytdot \rangle &= \langle \nabla f(\yt), \nabla f(\yt)^- \rangle = \norm{\nabla f(\yt)^-}_2^2
\end{align*}
Now we note that a coordinate of the projection $\pt_i$ can change for either one of the two reasons:
$\pt_i$ changes because $\xt_i$ increases, therefore $\ptdot_i=\xtdot_i\geq0$,
or $\pt_i$ changes because $\yt_i$ decreases, therefore $\ptdot_i=\ytdot_i\leq0$.
In the former case we have 
\begin{align*}
\nabla_i f(\pt) \ptdot_i & =\nabla_i f(\pt)\xtdot_i
\\
&\geq\nabla_i f(\yt)\xtdot_i =
\nabla_i f(\yt) \nabla_i f(\xt)^+
\\
&\geq
\nabla_i f(\yt)^- \nabla_i f(\xt)^+
\end{align*}
In the latter case we have
\begin{align*}
\nabla_i f(\pt)\ptdot_i &=\nabla_i f(\pt)\ytdot_i
\\
&\geq\nabla_i f(\xt)\ytdot_i = \nabla_i f(\xt) \nabla_i f(\yt)^-
\\
&\geq \nabla_i f(\xt)^{+}  \nabla_i f(\yt)^-
\end{align*}
Therefore $$\langle \nabla f(\pt), \ptdot \rangle \geq \langle \nabla f(\xt)^+, \nabla f(\yt)^- \rangle$$

Therefore we get that
\begin{align*}
&\frac{1}{2}\left(\langle \nabla f(\xt), \xtdot\rangle +
\langle \nabla f(\yt), \ytdot \rangle \right)
+\nabla \langle f(\pt), \ptdot \rangle
\\
&\geq
\frac{1}{2}\left(\norm{\nabla f(\xt)^{+}}_2^{2}
+\norm{\nabla f(\yt)^{-}}_2^{2}\right)
+\langle \nabla f(\xt)^{+}, \nabla f(\yt)^{-}\rangle \geq0
\end{align*}

\end{document}